\documentclass[journal,twoside,web]{ieeecolor}
\usepackage{lcsys}
\usepackage{cite}
\usepackage{amsmath,amssymb,amsfonts}
\usepackage{algorithmic}
\usepackage{graphicx}
\usepackage{textcomp}
\usepackage{xcolor}
\usepackage{hyperref}
\newcommand{\rev}[1]{{\color{black}#1}}   
\newtheorem{lemma}{Lemma}
\newtheorem{proposition}{Proposition}

\newtheorem{definition}{Definition}
\newtheorem{remark}{Remark}

\newcommand{\KL}{\operatorname{KL}}
\newcommand{\EE}{\mathbb{E}}
\newcommand{\Var}{\operatorname{Var}}
\newcommand{\Cov}{\operatorname{Cov}}
\newcommand{\RR}{\mathbb{R}}
\newcommand{\cN}{\mathcal{N}}
\newcommand{\cJ}{\mathcal{J}}

\def\BibTeX{{\rm B\kern-.05em{\sc i\kern-.025em b}\kern-.08em
    T\kern-.1667em\lower.7ex\hbox{E}\kern-.125emX}}
\begin{document}
\title{The Risk-Sensitive Schr\"odinger Bridge\\ Is Not a KL Projection}
\author{Hamidreza Behjoo  \IEEEmembership{}
\thanks{H.~Behjoo is with the Suzhou Institute of Systems Medicine,
        Chinese Academy of Medical Sciences, Suzhou 215123, China
        (e-mail: hamidreza.behjoo@gmail.com).} }

\maketitle
\thispagestyle{empty}
\begin{abstract}
	The Schr\"odinger bridge owes its computational power to a single structural
	fact: by Girsanov's theorem the controlled problem is a Kullback--Leibler (KL)
	projection onto a fixed reference measure, solvable by alternating projections.
	This letter shows that the fact does not survive risk sensitivity. When the
	expected path cost is replaced by the entropic risk measure and both endpoint
	marginals are kept as hard constraints, the resulting fixed-point
	bridge value $J_\theta$ (the soft-problem value at the multiplier that enforces the terminal constraint) admits no representation as a constrained KL minimum against any fixed path-space
	reference with a regular endpoint law (a class strictly larger than the
	uniformly elliptic diffusion references: no Markov property is required),
	even allowing \rev{a constant rescaling of the divergence and} an additive
	normalisation depending on the initial marginal. Moreover, no single
	reference generates the one-parameter family in the risk parameter. The obstruction is
	computed in closed form: the Gaussian bridge value violates, by exactly
	$\theta/2$, a heat equation that any Gaussian smoothing of a fixed endpoint
	density must obey. In place of the projection, the theory rests on a
	terminal-multiplier fixed point and an asymmetric factorisation penalising the score energy
	of the backward factor.
\end{abstract}


\section{Introduction}
\label{sec:intro}

The Schr\"odinger bridge problem lies at the intersection of stochastic
control, optimal transport, and generative
modelling~\cite{Leonard2014,chen2021liaisons}. Its
computational power rests on a single structural fact: by Girsanov's theorem
the controlled problem is a Kullback--Leibler (KL) projection onto a fixed
reference measure, and can therefore be solved by alternating projections or
Sinkhorn iteration in the static formulation. This letter shows that the projection
breaks when the expected path cost is replaced by the entropic risk measure
$\frac1\theta\log\EE\,e^{\theta\,\mathrm{cost}}$, $\theta\in(0,1)$, with
both endpoint marginals kept as hard constraints. The proof computes the
obstruction explicitly: the Gaussian bridge value violates, by exactly
$\theta/2$, a heat equation that the endpoint law of any regular reference
must obey (Section~\ref{sec:no-kl}).

\rev{The motivation for the entropic criterion is classical: the exponential
tilt charges all cumulants of the path cost, not only its mean, so the optimal
behaviour is shaped by rare, expensive
trajectories~\cite{jacobson1973,whittle1990}.}

Risk sensitivity alone would destroy nothing: with a fixed path cost
the entropic objective differs from a linear functional of the path measure
only by a monotone transformation. The obstruction
here is that the kinetic cost is a functional of the controlled law itself rather than a fixed path functional; jointly with the hard endpoint constraints, this destroys the KL-projection structure on which the classical
and computational theory rests (Section~\ref{sec:no-kl}). Two values must be
kept distinct throughout: the \emph{fixed-point bridge value} $J_\theta$ of
\eqref{eq:bridge-value}, developed here, and the \emph{hard-constrained
	risk-sensitive value} $\cJ_\theta$ of \eqref{eq:hard-rs}, a different object
for $\theta>0$. The two coincide at $\theta=0$ and separate at first order by
an explicit variance--covariance gap \eqref{eq:first-order-gap}.

\paragraph*{Contributions}
(i) a formulation of the risk-sensitive bridge in which the hard terminal
constraint is enforced by a terminal-multiplier fixed point, with a
gauge-invariant bridge value $J_\theta$ (Section~\ref{sec:rsb});
(ii) a three-step impossibility proof (the per-$\theta$ representation is circular, no reference is uniform in $\theta$, no endpoint-regular reference works at all), the obstruction computed
explicitly as a $\theta/2$ heat-equation defect and transferred to the
hard-constrained value for small $\theta$ (Section~\ref{sec:no-kl});
(iii) a closed-form Riccati reduction for Gaussian targets to decoupled scalar
shooting equations, proving well-posedness of the fixed point in that setting
(Section~\ref{sec:gaussian}); and
(iv) the structure that replaces the projection: an asymmetric factorisation
$p^*=\varphi\,\psi^{1/(1-\theta)}$ penalising the score energy of the backward
factor (Section~\ref{sec:replacement}).

\paragraph*{Related work}
Exponential-of-integral cost criteria go back to~\cite{jacobson1973} and were
developed in generality by~\cite{whittle1990,flemingmceneaney1995};
the minimax game against an entropy-penalised adversary is central to robust
control~\cite{basar1995}, and
risk-sensitive extensions within linearly solvable and path-integral control
appear in~\cite{vandenbroek2010}. All of these consider
unconstrained problems: the terminal law is free, and the two hard
marginal constraints responsible for our obstruction are absent. Pricing the
log-likelihood rather than the kinetic energy~\cite{ito2024renyi} yields a
R\'enyi divergence: convex, but again with a free terminal law. Existing ``robust'' bridge variants relax the hard terminal
constraint~\cite{chen2019relaxed} or address parameter perturbations; none
modifies the path-cost objective, and all retain the KL-projection structure.
\rev{To the best of our knowledge this is the first no-go characterisation of
fixed-reference KL duality for the bridge value, allowing general references
with regular endpoint laws, constant rescaling, and additive normalisation.
Concurrent work~\cite{cherian2026leqg} solves the hard-constrained Gaussian
case locally via coupled Riccati equations, reinforcing that the obstruction
is the KL route, not solvability of the hard problem.}

\section{Preliminaries: the standard bridge}
\label{sec:prelim}

Throughout, $\alpha:=1-\theta$, $\cN(m,\Sigma)$ denotes the Gaussian law, and
$\langle f,\mu\rangle:=\int f\,d\mu$\rev{; absolute continuity of marginals
is understood with respect to Lebesgue measure}. Fix the horizon $[0,1]$ and consider the controlled diffusion in $\RR^d$
\begin{equation}
	\label{eq:sde}
	dX_t = \bigl(f_t(X_t)+u_t(X_t)\bigr)\,dt+\sqrt{\kappa_t}\,dW_t,\qquad X_0\sim\mu_0,
\end{equation}
with given drift $f_t$, scalar noise intensity $\kappa_t>0$, and Markov control
$u$; \rev{$u$ is measurable with finite energy,
$\EE\int_0^1\kappa_t^{-1}\|u_t\|^2dt<\infty$, and $V$ is measurable and
bounded below.} $P^u$ \rev{denotes the path law of the controlled state
process} on $C([0,1];\RR^d)$; $P^0$ is the uncontrolled reference
\rev{($u\equiv0$)}. The
standard Schr\"odinger bridge~\cite{Leonard2014} is
\begin{equation}
	\label{eq:sb}
	\inf_{u}\;\EE^{P^u}\!\left[\int_0^1\!\left(\frac{1}{2\kappa_t}\|u_t(X_t)\|^2+V_t(X_t)\right)dt\right]
	\quad\text{s.t.}\;X_1\sim\mu_1.
\end{equation}
\rev{The first term of the integrand is the \emph{kinetic cost} (the control
energy), the second the potential cost.} By Girsanov's theorem,
\rev{$dP^u/dP^0=\exp\bigl(\int_0^1\kappa_t^{-1/2}u_t\cdot dW_t
-\frac12\int_0^1\kappa_t^{-1}\|u_t\|^2dt\bigr)$, so} the objective equals
$\KL(P^u\|P^V)$ up to an additive
constant, where $dP^V/dP^0\propto\exp(-\int_0^1 V_t(X_t)\,dt)$; the bridge is
therefore the static problem
$\inf_{Q:\,Q_0=\mu_0,\,Q_1=\mu_1}\KL(Q\|P^V)$, \rev{the \emph{I-projection}
(information projection: the KL-minimiser over the constraint set) of the fixed
reference $P^V$ onto the two-marginal constraint set}. \rev{The value function $S_0$ of the soft problem associated with
\eqref{eq:sb} (the same objective without the terminal constraint, but with an
added terminal cost $\lambda(X_1)$) satisfies}
\begin{equation}
	\label{eq:hjb-std}
	\begin{split}
		-\partial_t S_0 &= V_t+f_t\cdot\nabla S_0+\frac{\kappa_t}{2}\Delta S_0-\frac{\kappa_t}{2}\|\nabla S_0\|^2,\\
		S_0(1,x)&=\lambda(x),
	\end{split}
\end{equation}
with optimal feedback $u^*=-\kappa_t\nabla S_0$, and the Cole--Hopf transform
$S_0=-\log\psi$ linearises \eqref{eq:hjb-std} to
\begin{equation}
	\label{eq:psi-bwd}
	-\partial_t\psi+V\psi = f_t\cdot\nabla\psi+\frac{\kappa_t}{2}\Delta\psi,
	\qquad \psi(1,x)=e^{-\lambda(x)}.
\end{equation}
With the forward factor $\varphi$ solving
\begin{equation}
	\label{eq:phi-fwd}
	\partial_t\varphi = -\nabla\cdot(f_t\varphi)+\frac{\kappa_t}{2}\Delta\varphi-V\varphi,
	\qquad \varphi(0,\cdot)=\frac{\mu_0}{\psi(0,\cdot)},
\end{equation}
the optimally controlled density admits the symmetric factorisation
$p^*=\varphi\psi$. We work formally with classical positive solutions of the
forward--backward system under standard parabolic regularity; the Gaussian
case of Section~\ref{sec:gaussian} is fully self-contained.

\section{The risk-sensitive bridge}
\label{sec:rsb}

Replace the linear expectation in \eqref{eq:sb} by the cumulant-generating
function. For a prescribed terminal cost $\lambda$, the \emph{soft}
risk-sensitive value is
\begin{equation}
	\label{eq:rs-obj}
	\begin{split}
		S_\theta(\lambda):=\inf_{u}\frac{1}{\theta}
		\log\EE^{P^u}\!\Bigl[\exp\Bigl(&\theta\!\int_0^1\!\Bigl(\tfrac{\|u_t\|^2}{2\kappa_t}+V_t(X_t)\Bigr)dt\\
		&+\theta\lambda(X_1)\Bigr)\Bigr],
	\end{split}
\end{equation}
with only the initial constraint imposed; at $\theta=0$ this is the classical
soft value, and with the terminal constraint restored the bridge value
\eqref{eq:bridge-value} recovers \eqref{eq:sb} exactly. The \emph{bridge} problem imposes in addition the hard terminal
constraint $X_1\sim\mu_1$, and $\lambda$ is then no longer a datum: it becomes a
\emph{terminal multiplier}, i.e.\ a shooting parameter, enforcing the constraint.

\subsection{Gauge and the bridge value}
The soft value is exactly shift-equivariant,
$S_\theta(\lambda+c)=S_\theta(\lambda)+c$ for $c\in\RR$, while the terminal coupling
\eqref{eq:Psi1} below is blind to constants; the multiplier is determined only
up to the coset $\lambda^*+\RR$ \rev{(the \emph{gauge}: an additive constant)}, and the optimally controlled process is
untouched by the constant. The gauge-invariant combination is the difference
\begin{equation}
	\label{eq:bridge-value}
	J_\theta(\mu_0\to\mu_1):=S_\theta(\lambda^*)-\langle\lambda^*,\mu_1\rangle,
\end{equation}
with $\lambda^*$ any representative of the coset.

\begin{definition}[Bridge value]
	\label{def:bridge-value}
	The \emph{risk-sensitive bridge value} is the gauge-invariant fixed-point
	value \eqref{eq:bridge-value}. For computations we fix the representative by
	the gauge $\langle\lambda^*,\mu_1\rangle=0$, under which
	$J_\theta=S_\theta(\lambda^*)$; the definition does not depend on the gauge.
\end{definition}

The definition is canonical at $\theta=0$: the terminal
coupling becomes the classical Schr\"odinger system and
\eqref{eq:bridge-value} is precisely the Legendre identity
$J_0=\sup_{\lambda}\{S_0(\lambda)-\langle\lambda,\mu_1\rangle\}$, the classical bridge
value recovered with no arbitrary choices. For $\theta>0$ the Legendre
identity fails: the first variation of
$\lambda\mapsto S_\theta(\lambda)-\langle\lambda,\mu_1\rangle$ at the fixed point is
$q_1^\theta-\mu_1$, with $q_1^\theta$ a \emph{tilted} terminal law, so the fixed point is not a saddle point of a Lagrangian.
\rev{Saddle-point duality is nonetheless the wrong expectation to calibrate
on: $J_\theta$ is the value that any scheme keeping the terminal potential
inside the exponential delivers operationally -- penalty and
terminal-multiplier methods solve the soft problem \eqref{eq:rs-obj} and read
off the value at the enforcing multiplier -- and it is what the
generalised Cole--Hopf linearisation below computes. The obstruction of
Section~\ref{sec:no-kl} transfers to $\cJ_\theta$ itself for all sufficiently
small $\theta>0$, with a different first-order mechanism
(Proposition~\ref{prop:hard-value-obstruction}).}

\subsection{The hard value and the first-order gap}
For $\theta>0$, \eqref{eq:bridge-value} is not the value of the
hard-constrained risk-sensitive control problem
\begin{equation}
	\label{eq:hard-rs}
	\begin{split}
		\cJ_\theta(\mu_0\to\mu_1)
		:=\inf_{u:\,X_0\sim\mu_0,\,X_1\sim\mu_1}&\frac1\theta\log
		\EE^{u}\!\bigl[e^{\theta K^u}\bigr],\\
		K^u&=\int_0^1\Bigl(\tfrac{\|u_t\|^2}{2\kappa_t}+V_t\Bigr)dt .
	\end{split}
\end{equation}
\rev{(calligraphic $\cJ_\theta$ marks the hard value throughout.)}
The two coincide at $\theta=0$ and separate at first order: writing $u_0^*$
for the classical bridge control and $\lambda_0^*$ for the classical multiplier,
a cumulant expansion with an envelope argument gives
\begin{equation}
	\label{eq:first-order-gap}
	J_\theta-\cJ_\theta
	=\theta\Bigl(\tfrac12\Var^{u_0^*}\!\bigl(\lambda_0^*(X_1)\bigr)
	+\Cov^{u_0^*}\!\bigl(K,\,\lambda_0^*(X_1)\bigr)\Bigr)+o(\theta),
\end{equation}
generically nonzero, of either sign. The same expansion applied to
$J_\theta$ alone gives
\begin{equation}
	\label{eq:first-order-value}
	J_\theta=J_0+\frac{\theta}{2}\,
	\Var^{u_0^*}\!\bigl(K+\lambda_0^*(X_1)\bigr)+o(\theta),
\end{equation}
and \eqref{eq:first-order-gap} is its difference with the companion expansion
of $\cJ_\theta$. \rev{The expansion is formal: it requires the
cumulant-generating function of $K+\lambda_0^*(X_1)$ under the classical
bridge law to be finite in a neighbourhood of the origin, uniqueness of the
multiplier, and differentiability of $\theta\mapsto\lambda_\theta^*$ for the
envelope interchange. All three hold in the Gaussian setting of
Section~\ref{sec:gaussian}, where \eqref{eq:first-order-value} is confirmed by
the closed form; in the point-mass case at
$s^2=1$ it is even exact, without remainder: the classical control is the
constant $m$, so $\Var^{u_0^*}(K+\lambda_0^*(X_1))=m^2$ and
$v_\theta(m,1)=\frac{1+\theta}{2}m^2$.} The two problems differ because the kinetic cost is a
functional of the controlled law: in \eqref{eq:hard-rs} the KKT multiplier
acts outside the entropic exponential, while \eqref{eq:bridge-value} keeps
the terminal potential inside it -- precisely what makes the generalised
Cole--Hopf transform below available.

\subsection{HJB equation and generalised Cole--Hopf}
By dynamic programming, the state-conditioned soft value $S_\theta(t,x)$
satisfies
\begin{equation}
	\label{eq:hjb-rs}
	\begin{split}
		-\partial_t S_\theta &= V_t+f_t\cdot\nabla S_\theta+\frac{\kappa_t}{2}\Delta S_\theta
		-\frac{\kappa_t}{2}(1-\theta)\|\nabla S_\theta\|^2,\\
		S_\theta(1,x)&=\lambda(x).
	\end{split}
\end{equation}
The generalised Cole--Hopf transform
\begin{equation}
	\label{eq:ch-rs}
	\psi_t(x):=\exp\!\bigl(-(1-\theta)S_\theta(t,x)\bigr)
	\;\Longleftrightarrow\;
	S_\theta(t,x)=-\frac{\log\psi_t(x)}{1-\theta}
\end{equation}
linearises \eqref{eq:hjb-rs} to the backward PDE
\begin{equation}
	\label{eq:psi-rs}
	\begin{split}
		-\partial_t\psi_t+(1-\theta)V_t\psi_t
		&= f_t\cdot\nabla\psi_t+\frac{\kappa_t}{2}\Delta\psi_t,\\
		\psi(1,x)&=e^{-(1-\theta)\lambda(x)},
	\end{split}
\end{equation}
whose generator is $\theta$-independent; the optimal feedback is
\begin{equation}
	\label{eq:u-star}
	u_t^*(x)=\frac{\kappa_t}{1-\theta}\nabla\log\psi_t(x).
\end{equation}

\subsection{Asymmetric factorisation}
Under \eqref{eq:u-star}, the controlled density $p^*$ satisfies the
Fokker--Planck equation with drift
$f_t+\frac{\kappa_t}{1-\theta}\nabla\log\psi_t$.

\begin{proposition}[Forward decoupling]
	\label{prop:decouple}
	Let $\psi$ solve \eqref{eq:psi-rs} and let $p^*$ be the optimally controlled
	density. Then
	\begin{equation}
		\label{eq:factor-rs}
		p^*(t,x)=\varphi(t,x)\,\psi(t,x)^{\frac{1}{1-\theta}},
	\end{equation}
	where $\varphi$ satisfies the closed linear forward equation
	\begin{equation}
		\label{eq:G-fwd}
		\partial_t \varphi = -\nabla\cdot(f_t \varphi)+\frac{\kappa_t}{2}\Delta \varphi-V_{\rm eff}(t,x)\varphi,
	\end{equation}
	with effective potential
	\begin{equation}
		\label{eq:Veff}
		V_{\rm eff}(t,x)=V(t,x)+\frac{\kappa_t\theta}{2(1-\theta)^2}\|\nabla\log\psi(t,x)\|^2,
	\end{equation}
	and initial condition $\varphi(0,x)=\mu_0(x)/\psi(0,x)^{1/(1-\theta)}$.
\end{proposition}

\begin{proof}
	Insert the ansatz $p^*=\varphi\psi^r$ into the Fokker--Planck equation. The
	drift divergence produces a cross term proportional to
	$\nabla\varphi\cdot\nabla\psi$ with coefficient $r-(1-\theta)^{-1}$; setting
	it to zero yields $r=1/(1-\theta)$. Dividing by $\psi^r$ and using
	\eqref{eq:psi-rs} to eliminate $\partial_t\psi$ and $\Delta\psi$, the
	remaining scalar terms combine to give \eqref{eq:G-fwd}--\eqref{eq:Veff}.
\end{proof}

Since $\theta\ge0$, \eqref{eq:Veff} gives $V_{\rm eff}\ge V$ pointwise,
strictly wherever $\nabla\log\psi\neq0$: risk sensitivity can only add to the
potential landscape seen by the forward factor, concentrated where the
backward score is steep. The hard terminal constraint $p^*(1,\cdot)=\mu_1$
reads
\begin{equation}
	\label{eq:Psi1}
	\varphi(1,x)\,\psi(1,x)^{1/(1-\theta)}=\mu_1(x),
\end{equation}
equivalently $\varphi(1,\cdot)\,e^{-\lambda^*(\cdot)}=\mu_1$ in terms of the
multiplier: the compact identity that makes $\lambda^*$ a multiplier for the hard
constraint, determined implicitly as a shooting parameter. At $\theta=0$ the
score-energy term vanishes and \eqref{eq:G-fwd}--\eqref{eq:Psi1} reduce to
\eqref{eq:phi-fwd} and the symmetric factorisation.

\section{The fixed-point value has no KL dual}
\label{sec:no-kl}

For the standard bridge, Girsanov yields
$\EE^{P^u}\!\int_0^1(\frac{\|u_t\|^2}{2\kappa_t}+V_t)\,dt
=\KL(P^u\|P^V)+\text{const}$, so the bridge is the I-projection
$\inf_{Q:\,Q_0=\mu_0,\,Q_1=\mu_1}\KL(Q\|P^V)$. It is natural to ask whether
$J_\theta$ admits an analogous projection representation for $\theta>0$. It
does not, in a precise sense. \rev{(The statements below concern the
fixed-point value $J_\theta$; the hard value $\cJ_\theta$ admits no
linear-programming formulation for the independent reason that its exponent is
law-dependent.)} The structural obstruction is the asymmetric
factorisation \eqref{eq:factor-rs}: the Gibbs form
$dQ^*/dR\propto\varphi\cdot\psi$, with factors of equal weight, underlies the
static KL dual, and it is what $\theta>0$ breaks. Asymmetry alone is not a
proof: Lemma~\ref{lem:match} shows that a solution-dependent reference
reproduces the risk-sensitive bridge optimizer in
full; the impossibility is argued at the level of the value.

\rev{The obstruction is proved in three steps, stated as separate results
below; throughout, the problem data $(f,\kappa,V)$ are fixed, $\theta\in(0,1)$, and the terminal coupling \eqref{eq:Psi1} is assumed to
admit a solution $\lambda^*+\RR$ for the endpoint pairs under consideration, so
that $J_\theta$ is defined (this holds in the Gaussian setting of
Section~\ref{sec:gaussian}). First, for each fixed $\theta$, a Markov reference
whose I-projection value reproduces $J_\theta$ exists only if the reference is
built from the unknown solution $\nabla S_\theta$ itself: the per-$\theta$
representation is circular (Lemma~\ref{lem:match}). Second, no single Markov
reference represents $J_\theta$ as a constrained KL minimum for two distinct
values of $\theta$ (Proposition~\ref{prop:no-uniform}). Third, no fixed
endpoint-regular path-space reference $R_\theta$, Markovian or not, no
additive correction $C(\mu_0)$ depending only on the initial marginal, and no
constant rescaling of the divergence,
represents $J_\theta$ for all absolutely continuous marginals with full
support and finite second moment (Proposition~\ref{prop:no-kl}); endpoint regularity is a two-sided Gaussian
condition on the endpoint joint law (Definition~\ref{def:endpoint-regular}). The
last two steps are proved for the scalar Brownian data $f\equiv0$, $\kappa=1$,
$V\equiv0$ of Lemma~\ref{lem:value-pm}: the obstruction appears already in the
simplest setting, which precludes any representation serving general data.}

\subsection{The per-$\theta$ representation is circular}
Consider a candidate \rev{\emph{Markov reference} $P^{h,U}$ (the path law of
a Markov} drift $h$\rev{, weighted by the potential} $U$\rev{)}
(a sub-probability path measure when $U\neq0$, to which $\KL$ extends verbatim); relative to it, $\inf_Q\KL(Q\|P^{h,U})$ is itself
an ordinary bridge problem whose value function $S_0^{h,U}$ satisfies
\eqref{eq:hjb-std} with $(f,V)$ replaced by $(h,U)$.

\begin{lemma}[Matching condition]
	\label{lem:match}
	Fix the problem data $(f,\kappa,V)$, a terminal cost $\lambda$, and a Markov
	reference $(h,U)$ for which the two HJB equations admit classical solutions.
	Then $S_0^{h,U}(\cdot\,;\lambda)\equiv S_\theta(\cdot\,;\lambda)$ if and only if
	\begin{equation}
		\label{eq:match}
		\begin{split}
			U(t,x)=V(t,x)&+\bigl(f(t,x)-h(t,x)\bigr)\cdot\nabla S_\theta(t,x)\\
			&+\frac{\kappa_t\theta}{2}\|\nabla S_\theta(t,x)\|^2
		\end{split}
	\end{equation}
	for all $(t,x)$. When \eqref{eq:match} holds with $h=f$, i.e.\
	$U=V_{\rm eff}$ of \eqref{eq:Veff}, and $\lambda=\lambda^*$ implements
	$X_1\sim\mu_1$ via \eqref{eq:Psi1}, the I-projection of $P^{f,V_{\rm eff}}$
	onto $\{Q_0=\mu_0,\,Q_1=\mu_1\}$ coincides with the risk-sensitive bridge optimizer $P^{u^*}$
	itself, not merely with its value.
\end{lemma}

\begin{proof}
	Subtracting the two HJB equations and using $S_0\equiv S_\theta$ gives
	\eqref{eq:match}; conversely, under \eqref{eq:match} the function
	$S_\theta$ solves the reference HJB with the same terminal datum, and
	uniqueness gives $S_0\equiv S_\theta$. For the last
	statement, the reference bridge's optimal total drift is
	$h-\kappa_t\nabla S_0$; with $h=f$ and $S_0=S_\theta$ this is
	$f+\frac{\kappa_t}{1-\theta}\nabla\log\psi$, exactly \eqref{eq:u-star}.
	Same initial law and same feedback give the same path measure.
\end{proof}

So a per-$\theta$ representation exists only by building the reference out of
the unknown solution $\nabla S_\theta$ itself: it presupposes the answer and is
no duality in any useful sense\rev{; read positively, the matching condition
made constructive is a fixed-point iteration on the multiplier coset}.

\subsection{Two ingredients}
The remaining proofs rest on an explicit Gaussian value and a static
reduction.

\begin{lemma}[Point-mass Gaussian value]
	\label{lem:value-pm}
	Let $d=1$, $f\equiv0$, $V\equiv0$, $\kappa=1$, $\mu_0=\delta_0$,
	$\mu_1=\cN(m,s^2)$, and write $\alpha=1-\theta$. Then
	\begin{equation}
		\label{eq:value-pm}
		\begin{split}
			v_\theta(m,s^2)&:=J_\theta(\delta_0\to\cN(m,s^2))\\
			&=\frac{1}{2\alpha}\log(1+a_1)-\frac{a_1}{2\alpha}\bigl(m^2+s^2\bigr)\\
			&\qquad+\frac{m^2}{M}-\frac{\alpha B}{2M^2}\,m^2,
		\end{split}
	\end{equation}
	where $a_1$ is determined by the point-mass shooting map
	\begin{equation}
		\label{eq:shoot-pm}
		s^2=\int_0^1\bigl(1+a_1 \tau\bigr)^{-2/\alpha}\,d\tau,
	\end{equation}
	and $M:=\int_0^1(1+a_1\tau)^{-1/\alpha-1}\,d\tau$,
	$B:=\int_0^1(1+a_1\tau)^{-2}\,d\tau=1/(1+a_1)$. The map $a_1\mapsto s^2$ is
	a strictly decreasing bijection of $(-1,\infty)$ onto $(0,\infty)$, with
	$s^2=1\iff a_1=0$. At $\theta=0$, \eqref{eq:value-pm} collapses to
	$\frac12(m^2+s^2-1-\log s^2)=\KL(\cN(m,s^2)\,\|\,\cN(0,1))$, as it must.
\end{lemma}

\begin{proof}
	Set $k=0$ in the scalar Riccati system of Section~\ref{sec:gaussian}:
	$\dot A=A^2$, $\dot b=Ab$, whose solutions are rational in $t$; passing
	$\sigma(0)\downarrow0$ in the shooting equation \eqref{eq:shoot} yields
	\eqref{eq:shoot-pm}, and direct evaluation of
	\eqref{eq:bridge-value} gives \eqref{eq:value-pm}. Moreover
	$\partial s^2/\partial a_1=-\frac{2}{\alpha}\int_0^1\tau(1+a_1\tau)^{-2/\alpha-1}d\tau<0$
	is nonvanishing, so $a_1$ is a $C^1$ function of $s^2$ down to $s^2=1$,
	with derivative $-\alpha$ there.
\end{proof}

\begin{definition}[Endpoint regularity]
	\label{def:endpoint-regular}
	A probability measure $R$ on $C([0,1];\RR^d)$ is \emph{endpoint-regular} if
	its endpoint law $R_{01}:=\operatorname{Law}_R(X_0,X_1)$ \rev{(the joint endpoint law under $R$)} has a positive
	continuous density $r(x,y)$ with two-sided Gaussian bounds:
	there exist $c,\Lambda>0$ and a continuous $m_0$ of at most linear growth
	such that
	\begin{equation}
		\label{eq:endpoint-regular}
		c^{-1}e^{-\Lambda|y-m_0(x)|^2}\;\le\; r(x,y)\;\le\; c\,e^{-\Lambda^{-1}|y-m_0(x)|^2}
	\end{equation}
	for all $x,y\in\RR^d$. This holds for uniformly elliptic diffusions with
	bounded measurable drift and bounded potential by Aronson's
	two-sided estimates~\cite{aronson1967}, and for Ornstein--Uhlenbeck
	references explicitly.
\end{definition}

\begin{lemma}[Static reduction and the point-mass limit]
	\label{lem:static}
	Let $R$ be endpoint-regular and set
	$I_R(\mu_0,\mu_1):=\inf_{Q:\,Q_0=\mu_0,\,Q_1=\mu_1}\KL(Q\|R)$. Then:
	(i) \emph{(static reduction)}
	$I_R(\mu_0,\mu_1)=\inf_{\pi\in\Pi(\mu_0,\mu_1)}\KL(\pi\|R_{01})$, a static
	Schr\"odinger problem on $\RR^d\times\RR^d$;
	(ii) \emph{(point-mass limit)} if $\mu_0^\varepsilon=\cN(0,\varepsilon I_d)$
	and $\mu_1,\mu_1^*$ are absolutely continuous with finite second moment and
	finite entropy, then, with $\rho_0(\cdot):=R(X_1\in\cdot\mid X_0=0)$,
	\begin{equation}
		\label{eq:pm-limit}
		I_R(\mu_0^\varepsilon,\mu_1)-I_R(\mu_0^\varepsilon,\mu_1^*)
		\;\xrightarrow[\varepsilon\downarrow0]{}\;
		\KL(\mu_1\|\rho_0)-\KL(\mu_1^*\|\rho_0).
	\end{equation}
\end{lemma}

\begin{proof}
	Part (i) is the chain rule for KL divergence disintegrated over the
	endpoint pair: the bridge term is minimised at $0$ by adopting the bridges
	of $R$. \rev{Part (ii) is a two-sided squeeze under
	\eqref{eq:endpoint-regular}. Write $\rho_x(\cdot):=R(X_1\in\cdot\mid
	X_0=x)$. Upper bound: the product coupling
	$\mu_0^\varepsilon\otimes\mu_1$ is admissible, so
	$I_R(\mu_0^\varepsilon,\mu_1)\le\KL(\mu_0^\varepsilon\|R_0)
	+\EE_{x\sim\mu_0^\varepsilon}\KL(\mu_1\|\rho_x)$, and the expectation
	converges to $\KL(\mu_1\|\rho_0)$, since \eqref{eq:endpoint-regular}
	makes $\log\rho_x(\cdot)\to\log\rho_0(\cdot)$ dominated by
	$C(1+|\cdot|^2)$, integrable against $\mu_1$. Lower bound: for any
	coupling $\pi$ with second marginal $\mu_1$,
	$\EE_x\KL(\pi(\cdot\,|\,x)\|\rho_x)\ge\EE_x\KL(\pi(\cdot\,|\,x)\|\rho_0)
	-\EE_x\langle\pi(\cdot\,|\,x),\,|\log\rho_x-\log\rho_0|\rangle
	\ge\KL(\mu_1\|\rho_0)-o(1)$, by joint convexity of $\KL$ and the same
	domination. The common term $\KL(\mu_0^\varepsilon\|R_0)$ cancels in the
	difference \eqref{eq:pm-limit}.}
\end{proof}

\rev{Part (i) is the classical dynamic-to-static reduction
\cite{follmer1988,Leonard2014}: with $R^{x,y}:=R(\cdot\mid(X_0,X_1)=(x,y))$
the pinned (endpoint-conditioned) processes of the reference, every $Q\ll R$
satisfies
$\KL(Q\|R)=\KL(Q_{01}\|R_{01})+\int\KL(Q^{x,y}\|R^{x,y})\,dQ_{01}(x,y)$,
minimised by adopting the pinned processes of $R$; the dynamic bridge is the
static Schr\"odinger problem on $\RR^d\times\RR^d$ run with reference pinned
processes. For $\theta>0$ the reduction survives only in a solution-dependent
form: by Lemma~\ref{lem:match} the optimiser is the I-projection of
$P^{f,V_{\rm eff}}$, so its pinned processes are those of
$P^{f,V_{\rm eff}}$ and $V_{\rm eff}$ \eqref{eq:Veff} depends on the
target $\mu_1$ through $\psi$. Section~\ref{sec:gaussian} computes this
dependence explicitly.}

Some restriction on $C$ is necessary: if it could depend on both marginals,
any reference would represent $J_\theta$ via
$C(\mu_0,\mu_1):=J_\theta-\inf\KL(\cdot\,\|R_\theta)$. Dependence on $\mu_0$
alone is the natural class: for the standard bridge the additive term is the
Feynman--Kac normalisation, independent of the marginals altogether.

\subsection{No reference works uniformly in $\theta$}
\begin{proposition}[No $\theta$-uniform KL dual]
	\label{prop:no-uniform}
	There is no Markov reference $(h,U)$, chosen independently of $\theta$, with
	$h$ Lipschitz in $x$ of at most linear growth and $U$ bounded, and constants
	$C_\theta$ independent of the marginals, such that, in the setting of
	Lemma~\ref{lem:value-pm},
	$J_\theta(\mu_0\to\mu_1)
	=\inf_{Q:\,Q_0=\mu_0,\,Q_1=\mu_1}\KL(Q\|P^{h,U})+C_\theta$
	holds for two distinct values $\theta\in[0,1)$ \rev{(one of which may be
	the classical $\theta=0$)} and all absolutely continuous
	marginals $\mu_0,\mu_1$.
\end{proposition}

\begin{proof}
	Suppose a fixed pair $(h,U)$ satisfies the identity for
	$\theta_1\neq\theta_2$; write $R=P^{h,U}$. Apply the identity with
	$\mu_0^\varepsilon=\cN(0,\varepsilon)$ and two Gaussian targets, subtract,
	and let $\varepsilon\downarrow0$. On the left the Riccati formulas of
	Lemma~\ref{lem:value-pm} are continuous in the initial variance down to
	$\sigma(0)=0$; on the right Lemma~\ref{lem:static} applies, since the
	endpoint kernel of $P^{h,U}$ is endpoint-regular \rev{(bounded $h$ by Aronson's
	estimate~\cite{aronson1967}; linear $h$ explicitly Gaussian; linear-growth
	Lipschitz $h$ by Girsanov comparison against the driftless kernel:
	the Girsanov exponential is a true martingale under linear growth by
	Bene\v{s}' condition~\cite{benes1971} -- Novikov applied on
	sufficiently short subintervals, Fernique's estimate supplying the
	exponential moments on each -- and Cauchy--Schwarz against the
	Gaussian tail of the Brownian bridge then transfers the two-sided
	bounds to the controlled kernel, with $m_0(x)$ the controlled mean)}. Hence, with
	$\rho_0(\cdot):=R(X_1\in\cdot\mid X_0=0)$,
	$v_{\theta_i}(m,s^2)=\KL(\cN(m,s^2)\,\|\,\rho_0)+C_{\theta_i}$ for $i=1,2$
	and all $m\in\RR$, $s^2>0$. Subtracting, $v_{\theta_1}-v_{\theta_2}$ must
	be the constant $C_{\theta_1}-C_{\theta_2}$. But at $s^2=1$ (so $a_1=0$, $M=B=1$),
	\eqref{eq:value-pm} gives
	$v_\theta(m,1)=m^2-\frac{\alpha}{2}m^2=\frac{1+\theta}{2}m^2$, hence
	$v_{\theta_1}(m,1)-v_{\theta_2}(m,1)=\tfrac12(\theta_1-\theta_2)m^2$, which
	is not constant in $m$, a contradiction.
\end{proof}

Thus the risk-sensitive bridge is a family of distinct control problems, not a
re-weighting of one reference measure.

\subsection{No fixed reference works at all}
\begin{proposition}[No endpoint-regular reference]
	\label{prop:no-kl}
	Let $d=1$, $f\equiv0$, $\kappa=1$, $V\equiv0$, and let $\theta\in(0,1)$.
	There exists no endpoint-regular probability measure $R_\theta$ on
	$C([0,1];\RR)$, no map $\mu_0\mapsto C(\mu_0)$\rev{, and no constant
	rescaling $\beta>0$} such that
	$J_\theta(\mu_0\to\mu_1)=\rev{\beta}\inf_{Q:\,Q_0=\mu_0,\,Q_1=\mu_1}\KL(Q\|R_\theta)+C(\mu_0)$
	for all absolutely continuous marginals $\mu_0,\mu_1$ with full support and
	finite second moment.
\end{proposition}

\begin{proof}
	Suppose $(R_\theta,C)$ satisfies the identity. \emph{Step 1: point-mass
		test.} Apply the identity with $\mu_0^\varepsilon=\cN(0,\varepsilon)$ and
	two Gaussian targets, and subtract: $C$ cancels since it depends only on
	$\mu_0$. Let $\varepsilon\downarrow0$. On the left, the Riccati formulas
	are continuous down to $\sigma(0)=0$; on the right,
	Lemma~\ref{lem:static}(ii) applies with
	$\rho_0(\cdot):=R_\theta(X_1\in\cdot\mid X_0=0)$, giving
	\begin{equation}
		\label{eq:kl-form}
		v_\theta(m,s^2)=\rev{\beta}\,\KL\bigl(\cN(m,s^2)\,\|\,\rho_0\bigr)+\mathrm{const}
	\end{equation}
	for all $m\in\RR$, $s^2>0$, with the constant independent of $(m,s^2)$. \emph{Step 2: a deconvolution
		criterion.} Since
	$\KL(\cN(m,s^2)\|\rho)=-\tfrac12\log(2\pi s^2)-\tfrac12
	-\EE_{Y\sim\cN(m,s^2)}[\log\rho(Y)]$, the identity \eqref{eq:kl-form} says
	that
	$L_\theta(m,s^2):=\rev{-v_\theta(m,s^2)/\beta}-\tfrac12\log(2\pi s^2)-\tfrac12$
	is, up to a constant, the Gaussian smoothing at bandwidth $s^2$ of the
	fixed function $\log\rho_0$. Gaussian smoothing obeys the heat
	equation in $(m,s^2)$:
	\begin{equation}
		\label{eq:heat}
		\partial_{s^2}L_\theta=\tfrac12\,\partial_m^2 L_\theta
		\qquad\text{for all }m\in\RR,\ s^2>0,
	\end{equation}
	a necessary condition on the value $v_\theta$ alone. \emph{Step 3: explicit
		failure near $s^2=1$.} At $s^2=1$ one has $a_1=0$, $M=B=1$, so
	$v_\theta(m,1)=\frac{1+\theta}{2}m^2$ and, \rev{for $\beta=1$,}
	$\tfrac12\partial_m^2
	L_\theta(m,1)=-\tfrac12\partial_m^2 v_\theta(m,1)=-\frac{1+\theta}{2}$.
	On the other hand $v_\theta(0,s^2)=\frac{1}{2\alpha}\log(1+a_1)-\frac{a_1
		s^2}{2\alpha}$, and differentiating \eqref{eq:shoot-pm} at $a_1=0$ gives
	$\frac{da_1}{ds^2}\big|_{s^2=1}=-\alpha$; the two terms of
	$\partial_{s^2}v_\theta(0,s^2)$ then cancel at $s^2=1$, giving
	$\partial_{s^2}L_\theta|_{(0,1)}=-\frac12$. \rev{For $\beta=1$} the heat
	equation \eqref{eq:heat} would require
	$-\tfrac12=-\tfrac{1+\theta}{2}$, which fails for every $\theta>0$: the
	heat-equation defect is exactly $\theta/2$. \rev{The single-point test is
	blind to the rescaling -- at $\beta=1+\theta$ the two sides agree at
	$(0,1)$ -- but the obstruction survives away from it. Write
	$v_\theta(m,s^2)=c_2(s^2)\,m^2+c_0(s^2)$, quadratic in $m$ by
	\eqref{eq:value-pm}; \eqref{eq:heat} applied to $L_\theta$ reads
	$-c_2'(s^2)m^2-c_0'(s^2)-\beta/(2s^2)=-c_2(s^2)$ for all $m$, forcing
	$c_2'\equiv0$ whatever $\beta$: the $m^2$-coefficient of the value must
	be independent of the target variance. From
	\eqref{eq:value-pm}--\eqref{eq:shoot-pm},
	$c_2(s^2)=-\frac{a_1}{2\alpha}+\frac1M-\frac{\alpha B}{2M^2}$ with
	$a_1=a_1(s^2)$; this is stationary at $s^2=1$ and expands as
	\begin{equation}
		\label{eq:c2}
		c_2(s^2)=\frac{1+\theta}{2}+\frac{\theta^2(1+\theta)}{24}\,(s^2-1)^2
		+O\bigl((s^2-1)^3\bigr),
	\end{equation}
	with the quadratic coefficient strictly positive for every $\theta>0$:
	$c_2$ is not constant, and no choice of $(\beta,R_\theta,C)$ satisfies
	the identity.} Hence no density $\rho_0$, and
	\emph{a fortiori} no reference $R_\theta$, can satisfy \eqref{eq:kl-form}.
	At $\theta=0$ the two sides agree and \eqref{eq:kl-form} holds \rev{with
	$\beta=1$ and} $\rho_0=\cN(0,1)$, the Brownian endpoint law: the classical KL duality, recovered as a sanity check.
\end{proof}

\begin{proposition}[\rev{Local obstruction for the hard value}]
	\label{prop:hard-value-obstruction}
	\rev{Let $d=1$, $f\equiv0$, $\kappa=1$, $V\equiv0$. There exists
	$\theta_0>0$ such that for every $\theta\in(0,\theta_0)$ there is no
	endpoint-regular $R_\theta$ on $C([0,1];\RR)$, no $\beta>0$, and no map
	$\mu_0\mapsto C(\mu_0)$ with
	$\cJ_\theta(\mu_0\to\mu_1)=\beta\inf_{Q:\,Q_0=\mu_0,\,Q_1=\mu_1}
	\KL(Q\|R_\theta)+C(\mu_0)$ for all endpoint data. The argument applies to
	each fixed $\theta$ separately, so per-$\theta$ references are excluded as
	well.}
\end{proposition}

\begin{proof}[Proof sketch]
	\rev{Pythagoras around the classical bridge law $P^*$ gives
	$E_QK_Q=J_0+\KL(Q\|P^*)$ for every endpoint-constrained $Q$; inserting
	$u_0^*$ bounds $\cJ_\theta$ from above, and truncating the integrand of
	$K$ at $M=\theta^{-1/8}$ with the entropy inequality, Pinsker, and
	$e^x-1-x\ge\frac{x^2}{2}(1-|x|/3)$ forces every near-optimal control
	within $O(\theta)$ of $u_0^*$ in drift energy, proving
	$\cJ_\theta=J_0+\frac{\theta}{2}\Var^{u_0^*}(K)+o(\theta)$ with no
	envelope assumption. In the point-mass test
	$\Var^{u_0^*}(K)=m^2V_1(s^2)+V_2(s^2)$ with
	$V_1(s^2)=\frac{(s^2-1)^2}{3}+O((s^2-1)^3)$, positive for $s^2\neq1$. A
	represented value would satisfy
	$g_\theta(m,s_2^2)=\EE\,g_\theta(m+\sqrt{\Delta}\,Z,s_1^2)
	-\frac{\beta}{2}\log(s_2^2/s_1^2)$ for every $s_1^2<s_2^2$,
	$\Delta=s_2^2-s_1^2$; since $(g_\theta-g_0)/\theta\to\frac12(m^2V_1+V_2)$,
	$m$-independence of the right side forces $V_1$ constant, contradicting
	$V_1(1)=0$ and $V_1(2)=3-4\log2>0$. The full proof appears in the
	extended version of this letter.}
\end{proof}

\begin{remark}[\rev{Computational content}]
	\label{rem:computational}
	\rev{The lost duality is algorithmic, not a hardness result. At $\theta=0$
	the value is the concave maximisation
	$\sup_\lambda\{S_0(\lambda)-\langle\lambda,\mu_1\rangle\}$, computable by
	Sinkhorn iteration. For $\theta>0$ the fixed point is no saddle point of a
	Lagrangian, so no such convex iteration is available; what replaces it is a
	shooting problem on the multiplier, decoupled into scalar strictly monotone
	equations with unique roots in the Gaussian case
	(Section~\ref{sec:gaussian}). Where precisely the classical alternation
	breaks: Sinkhorn's iteration is a pair of I-projections against a
	\emph{fixed} kernel, contractive in Hilbert's projective
	metric~\cite{chen2021liaisons,delmoral2026sinkhorn};
	for $\theta>0$ the forward factor carries $V_{\rm eff}$
	\eqref{eq:Veff}, which depends on the current $\psi$, so the two
	half-updates no longer share a fixed kernel and the contraction argument
	loses its invariant at exactly this point. The loss of the KL projection
	should therefore not be read as the loss of every convergent alternating
	scheme: the scalar shooting map of Section~\ref{sec:gaussian} is one object that
	survives (strictly monotone, globally soluble), and whether a stable
	alternating scheme exists beyond the Gaussian case is open
	(Section~\ref{sec:conclusion}).}
\end{remark}

\begin{remark}[An inf--sup in place of the KL dual]
	\label{rem:dv}
	The structural reason is visible in the Donsker--Varadhan
	formula~\cite{donsker1975}: the entropic risk is itself
	a Legendre object, so $J_\theta$ is an inf--sup of KL-type objects,
	whereas the standard bridge is a single inf; a projection would force the
	inner supremum to collapse, which happens only at $\theta=0$.
\end{remark}

\section{The Gaussian case: closed-form Riccati reduction}
\label{sec:gaussian}

Specialise to $f_t=0$, $\kappa_t=\kappa>0$, $V(x)=\frac{k}{2}\|x\|^2$
($k\ge0$), $\mu_0=\cN(0,I_d)$, $\mu_1=\cN(m_1,\Sigma_1)$ with
$\Sigma_1\in\mathbb{S}_{++}^d$, and define
$\gamma:=\sqrt{\alpha k/\kappa}$, $\omega:=\sqrt{\alpha k\kappa}$. Posit the
log-quadratic ansatz
$\psi(t,x)=\exp(-\frac12 x^\top A(t)x+b(t)^\top x+c(t))$,
$\varphi(t,x)=\exp(-\frac12 x^\top D(t)x+E(t)^\top x+F(t))$ with
$A,D\in\mathbb{S}^d$. The optimally controlled marginal is then Gaussian
with precision $\Sigma(t)^{-1}=D(t)+A(t)/\alpha$; positivity is required only
of this sum.
Inserting the ansatz into \eqref{eq:psi-rs} gives the backward Riccati system
\begin{equation}
	\label{eq:riccati-A}
	\dot A=\kappa A^2-\alpha k I,\qquad
	\dot b=\kappa Ab,
\end{equation}
with $\dot c=\frac{\kappa}{2}(\operatorname{tr}A-\|b\|^2)$.
and since $\|\nabla\log\psi\|^2$ is quadratic, the effective potential
\eqref{eq:Veff} remains quadratic: the ansatz is self-consistent. The forward
equation \eqref{eq:G-fwd} yields
\begin{equation}
	\label{eq:riccati-D}
	\dot D=-\kappa D^2+kI+\frac{\kappa\theta}{\alpha^2}A^2,\qquad
	\dot E=-\kappa DE+\frac{\kappa\theta}{\alpha^2}Ab,
\end{equation}
and in the total precision the $kI$ terms cancel:
\begin{equation}
	\label{eq:riccati-Lambda}
	\frac{d}{dt}\Sigma(t)^{-1}=-\kappa D^2+\frac{\kappa}{\alpha^2}A^2,
\end{equation}
a closed system for $(A,\Sigma^{-1})$, which recovers
the covariance equation
\begin{equation}
	\label{eq:cov-ode}
	\dot\Sigma=\kappa I-\frac{\kappa}{\alpha}\bigl(A\Sigma+\Sigma A\bigr),
\end{equation}
exactly the moment evolution of the controlled Ornstein--Uhlenbeck SDE with
drift $-\frac{\kappa}{\alpha}(A(t)x-b(t))$.

In the eigenbasis of $\Sigma_1$, $A(t)$ and $\Sigma(t)$ are simultaneously
diagonal; each diagonal entry satisfies the scalar system
\begin{equation}
	\label{eq:scalar}
	\dot a_i=\kappa a_i^2-\alpha k,\qquad
	\dot\sigma_i=\kappa-\frac{2\kappa}{\alpha}a_i\sigma_i,
\end{equation}
with $\sigma_i(0)=1$, $\sigma_i(1)=\ell_i$. The backward equation has an
explicit hyperbolic solution, and the terminal condition reduces to the scalar
shooting equation
\begin{equation}
	\label{eq:shoot}
	\ell_i=\gamma^{2/\alpha}\left(g_i(1)^{-2/\alpha}
	+\kappa\int_0^1 g_i(\tau)^{-2/\alpha}\,d\tau\right),
\end{equation}
with $g_i(\tau)=\gamma\cosh(\omega\tau)+a_{i,1}\sinh(\omega\tau)$ and
$a_{i,1}:=a_i(1)$: one transcendental equation per coordinate. The mean
condition is then affine per coordinate for $b_{i,1}$.

\begin{remark}[Well-posedness; Brownian limit]
	\label{rem:shoot-wp}
	The map $a_{i,1}\mapsto\ell_i$ in \eqref{eq:shoot} is strictly decreasing (for $\tau>0$, $g_i(\tau)$ is strictly increasing in $a_{i,1}$) and surjective onto $(0,\infty)$, so a unique root exists for
	every $\ell_i>0$. This
	proves existence and uniqueness of the terminal-multiplier coset of
	Definition~\ref{def:bridge-value} in the Gaussian setting. In the Brownian
	case $k=0$ the hyperbolic formula degenerates to a rational one, and
	passing $\sigma_i(0)\downarrow0$ recovers the point-mass map
	\eqref{eq:shoot-pm} used in Section~\ref{sec:no-kl}.
\end{remark}

\begin{remark}[\rev{Classical limit; Gaussian closed forms}]
	\label{rem:gauss-lit}
	\rev{At $\theta=0$ the system \eqref{eq:riccati-A}--\eqref{eq:shoot}
	reduces to the known Gaussian bridge: the Riccati pair is that of
	covariance steering~\cite{chen2016steering}, and the endpoint coupling
	is the static Gaussian entropic optimal transport problem, solved in
	closed form in~\cite{bunne2023gaussian,janati2020unbalanced,mallasto2022entropy}
	and, for linear--Gaussian references (the case $k>0$ here),
	in \cite{akyildiz2026gaussian}. For $\theta>0$ the backward
	equation \eqref{eq:riccati-A} is the classical Riccati equation for the
	scaled data $(\kappa,\alpha k)$, and the deformation sits in
	\eqref{eq:riccati-D} through the score-energy forcing
	$\kappa\theta\alpha^{-2}A^2$: risk sensitivity keeps the Riccati closure
	and loses the symmetric Gibbs form $p^*=\varphi\psi$ on which Sinkhorn
	iteration rests. At $\theta=0$ the endpoint covariance also admits the
	factorised fixed-point form of~\cite{delmoral2026bridging}, globally
	Lipschitz in inverse coordinates and continuous through the point-mass
	limit of Lemma~\ref{lem:value-pm}. A quantitative stability analogue for
	$\theta>0$ is open. In independent concurrent
	work~\cite{cherian2026leqg} the multivariate matched-channel problem is
	solved locally near $\theta=0$; the scalar shooting map of
	Remark~\ref{rem:shoot-wp} gives global well-posedness for every
	$\theta<1$.}
\end{remark}

\paragraph*{The obstruction, in factor language}
In the same Brownian setting, write the terminal forward factor as
$\varphi(1,x)\propto\exp(-\tfrac12 D(1)x^2+E(1)x)$. Precision addition in the
coupling \eqref{eq:Psi1} and the shooting map \eqref{eq:shoot-pm} give
$D(1)=\frac{1}{s^2}-\frac{a_1}{\alpha}
=1+\frac{\theta}{3(1-\theta)^2}a_1^2+O(a_1^3)$. For $\theta>0$ the terminal
precision from a point mass varies with the target variance near $s^2=1$,
whereas a fixed reference forces target-independent endpoint bridges
($\theta=0$: $D(1)\equiv1$). This is the factor-level form of the value-level
defect of Proposition~\ref{prop:no-kl}.

\begin{remark}[\rev{Optimiser-level obstruction}]
	\label{rem:pinned}
	\rev{The same dependence is visible on path space. In the setting of
	Lemma~\ref{lem:value-pm} the optimiser is Gauss--Markov with drift
	$(b(t)-a(t)x)/\alpha$; for a centred target its bridge pinned at
	$(0,y)$ is Gaussian with mean
	\begin{equation}
		\label{eq:pinned}
		m_y(t)=y\;w(t)^{1/\alpha}\,
		\frac{\int_0^t w(s)^{-2/\alpha}\,ds}{\int_0^1 w(s)^{-2/\alpha}\,ds},
	\end{equation}
	with $w(t):=1+a_1(1-t)$, $a_1=a_1(s^2)$ the shooting root of \eqref{eq:shoot-pm}. At
	$\theta=0$ \eqref{eq:pinned} is $ty$ for every target variance -- the
	Brownian bridge, as the static reduction of Lemma~\ref{lem:static}
	prescribes. For $\theta>0$ it depends on the target: the slope at $t=0$
	is $(1+a_1)^{-1/\alpha}/s^2=1-\frac{\theta}{6\alpha^2}a_1^2+O(a_1^3)$,
	equal to $1$ identically only at $\theta=0$.
	The I-projection of any reference $R$ adopts the pinned processes of $R$
	(Lemma~\ref{lem:static}(i)), so no single $R$ realises the optimiser for
	two targets of different variances: an optimiser-level counterpart of
	Proposition~\ref{prop:no-kl}.}
\end{remark}

\begin{figure}[t]
	\centering
	\includegraphics[width=\columnwidth]{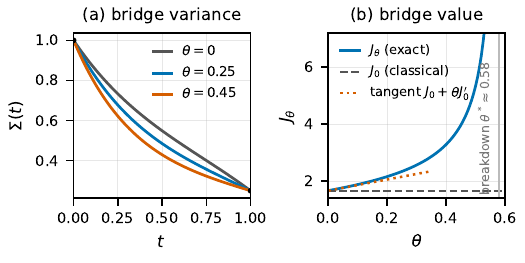}
	\caption{Gaussian bridge with $\kappa=k=1$, $\mu_0=\cN(0,1)$,
	$\mu_1=\cN(1,1/4)$, computed from \eqref{eq:scalar}--\eqref{eq:shoot}.
	(a) Marginal variance $\Sigma(t)$: endpoints matched exactly for every
	$\theta$; intermediate variance contracts as $\theta$ grows. (b) Bridge
	value $J_\theta$ leaving $J_0$ with the first-order slope of
	\eqref{eq:first-order-value}, diverging at \rev{the marked breakdown
	$\theta^*\approx0.58$ (exact root of $\theta a(0)/(1-\theta)=1$)}, the
	classical entropic breakdown~\cite{whittle1990}; \rev{point-mass initial
	laws stay finite for every $\theta<1$ (Lemma~\ref{lem:value-pm})}. At
	$\theta=0$ the curve coincides with the static Gaussian entropic-OT
	closed form to numerical precision.}
	\label{fig:gaussian}
\end{figure}

\paragraph*{Numerical illustration}
Figure~\ref{fig:gaussian} shows the one-dimensional case $\kappa=k=1$,
$\mu_0=\cN(0,1)$, $\mu_1=\cN(1,1/4)$, solved through the scalar system
\eqref{eq:scalar}--\eqref{eq:shoot}. The terminal constraint is met exactly
at every $\theta$, while the intermediate variance contracts as $\theta$
grows. The bridge
value leaves the classical value with the first-order slope of
\eqref{eq:first-order-value}, then grows superlinearly and diverges at a finite
$\theta$: for an unbounded initial law the entropic criterion has a finite
breakdown point, beyond which no finite-cost bridge exists. The breakdown
attaches to the unbounded initial law rather than to the construction: the
state-conditioned soft value $S_\theta(t,x)$ and the point-mass case of
Lemma~\ref{lem:value-pm} remain finite for every $\theta<1$.

\section{What replaces the projection}
\label{sec:replacement}

The failure of the KL projection is not the failure of structure. For
$\theta>0$ the fixed-point bridge retains the generalised Cole--Hopf
linearisation \eqref{eq:ch-rs}--\eqref{eq:psi-rs}, whose backward generator is
$\theta$-independent, and the asymmetric factorisation
$p^*=\varphi\,\psi^{1/(1-\theta)}$ under the effective potential
\eqref{eq:Veff}; the value remains an inf--sup of KL-type objects
(Remark~\ref{rem:dv}). What is lost is only the single fixed reference and the
projection geometry it induces. The mechanism is visible at first
order in $\theta$: expanding
$\psi_\theta=\psi^{(0)}+\theta\psi^{(1)}+O(\theta^2)$,
$\varphi_\theta=\varphi^{(0)}+\theta\varphi^{(1)}+O(\theta^2)$ around the standard
bridge pair and matching the $O(\theta)$ terms of \eqref{eq:psi-rs} and
\eqref{eq:G-fwd} gives linear equations with explicit sources: $\psi^{(1)}$
solves the standard backward PDE with source $V\psi^{(0)}$, and $\varphi^{(1)}$
solves the standard forward PDE with the negative source
$-\frac{\kappa_t}{2}\|\nabla\log\psi^{(0)}\|^2\,\varphi^{(0)}$.

\begin{proposition}[First-order density correction]
	\label{prop:first-order}
	The optimally controlled density expands as
	\begin{equation}
		\label{eq:p1}
		\begin{split}
			p_\theta^*(t,x)&=\varphi^{(0)}\psi^{(0)}
			+\theta\bigl[\varphi^{(0)}\psi^{(0)}\log\psi^{(0)}\\
			&\qquad+\varphi^{(0)}\psi^{(1)}+\varphi^{(1)}\psi^{(0)}\bigr]
			+O(\theta^2),
		\end{split}
	\end{equation}
	and, writing $p_0:=\varphi^{(0)}\psi^{(0)}$ and $\chi:=\psi^{(1)}/\psi^{(0)}$, the
	correction $p_1$ satisfies the linearised Fokker--Planck equation
	\begin{equation}
		\label{eq:p1-cons}
		\begin{split}
			\partial_t p_1&=-\nabla\cdot\bigl(\bigl(f+\kappa_t\nabla\log\psi^{(0)}\bigr)p_1\bigr)
			+\tfrac{\kappa_t}{2}\Delta p_1\\
			&\quad-\kappa_t\nabla\cdot\bigl(\bigl(\nabla\log\psi^{(0)}+\nabla\chi\bigr)p_0\bigr),
		\end{split}
	\end{equation}
	with $\int p_1(t,x)\,dx=0$ for all $t$.
\end{proposition}

\begin{proof}
	With
	$\nabla\log\psi_\theta=\nabla\log\psi^{(0)}+\theta\nabla\chi+O(\theta^2)$,
	matching the $O(\theta)$ terms of the exact Fokker--Planck equation gives
	\eqref{eq:p1-cons}; the integral constraint follows by differentiating
	$\int p_\theta^*\,dx=1$ at $\theta=0$.
\end{proof}

The first-order effect of risk sensitivity is therefore exactly localised:
forward mass is depleted where the \emph{score energy}
$\mathcal I[\psi^{(0)}]:=\|\nabla\log\psi^{(0)}\|^2$ of the standard bridge
is large (saddles, ridges, ambiguous zones between modes), so the ensemble
commits to a basin earlier: a design handle on which trajectories a
generative sampler avoids.

\section{Conclusion}
\label{sec:conclusion}

We have constructed a risk-sensitive deformation of the Schr\"odinger bridge
that preserves exact terminal matching and the linearising Cole--Hopf
structure, and proved that its value admits no fixed-reference KL dual: the
obstruction is an explicit heat-equation defect, ruling out every
endpoint-regular reference, uniformly and per $\theta$; for small $\theta$ it
attaches to the hard value as well
(Proposition~\ref{prop:hard-value-obstruction}). Open problems include a
general existence theory, contraction of the fixed-point iteration, the limit
$\theta\to1$, and duality for $\cJ_\theta$ beyond this regime.

\bibliographystyle{IEEEtran}
\bibliography{references}

\end{document}